\documentclass[11pt]{article}

\usepackage{amsmath, amsthm, amssymb, bbm, graphicx, enumerate, fullpage, url}
\usepackage{algorithmicx}
\usepackage{color, hyperref}

\hypersetup{
	colorlinks=true,
	linkcolor=blue,
	citecolor=magenta
}

\theoremstyle{plain}
\newtheorem{theorem}{Theorem}[section]
\newtheorem{proposition}[theorem]{Proposition}
\newtheorem{lemma}[theorem]{Lemma}
\newtheorem{corollary}[theorem]{Corollary}

\theoremstyle{definition}
\newtheorem{definition}[theorem]{Definition}

\newtheorem{example}[theorem]{Example}

\newcommand{\E}{\mathbb{E}}
\newcommand{\F}{\mathbb{F}}

\newcommand{\N}{\mathbb{N}}

\newcommand{\1}{\mathbbm{1}}

\newcommand{\calA}{\mathcal{A}}

\newcommand{\calC}{\mathcal{C}}
\newcommand{\calL}{\mathcal{L}}

\newcommand{\Herm}{{\rm Herm}}
\newcommand{\RS}{{\rm RS}}

\newcommand{\ceil}[1]{\left\lceil #1 \right\rceil}
\DeclareMathOperator{\id}{id}
\DeclareMathOperator{\lift}{Lift}
\renewcommand{\span}{{\rm span}}

\begin{document}

\title{High rate locally correctable codes via
lifting}

\author{Alan Guo
\thanks{CSAIL, Massachusetts Institute of
Technology, 32 Vassar Street, Cambridge, MA, USA. {\tt aguo@mit.edu}. Research
supported in part by NSF grants CCF-0829672, CCF-1065125,
and CCF-6922462, and an NSF Graduate Research Fellowship}}
\date{}
\maketitle

\begin{abstract}
We present a general framework for constructing high rate error
correcting codes
that are locally correctable (and hence locally decodable if linear)
with a sublinear
number of queries, based on
lifting codes with respect to functions on the coordinates.
Our approach generalizes the lifting of affine-invariant codes of
Guo, Kopparty, and Sudan and its generalization automorphic lifting,
suggested by Ben-Sasson et al, which lifts algebraic geometry codes
with respect to a group of automorphisms of the code.
Our notion of lifting is a natural alternative
to the degree-lifting of Ben-Sasson et al and
it carries two advantages. First, it overcomes the rate barrier inherent
in degree-lifting. Second, it is extremely flexible, requiring no special
properties (e.g.\ linearity, invariance) of the base code, and requiring
very little structure on the set of functions on the coordinates of the
code.

As an application, we construct new explicit families of
locally correctable codes
by lifting algebraic geometry
codes. Like the multiplicity codes of Kopparty, Saraf, Yekhanin 
and the affine-lifted
codes of Guo, Kopparty, Sudan, our codes of block-length $N$ can achieve
$N^\epsilon$ query complexity and $1-\alpha$ rate for any given $\epsilon,
\alpha > 0$ while correcting a constant fraction of errors, in contrast
to the Reed-Muller codes and the degree-lifted AG codes of Ben-Sasson et al
which face a rate barrier of $\epsilon^{O(1/\epsilon)}$. However, like the
degree-lifted AG codes, our codes are over an alphabet significantly smaller
than that obtained by Reed-Muller codes, affine-lifted codes, and multiplicity
codes.

Keywords: Error correcting codes, sublinear time decoding, algebraic geometry
codes, lifting
\end{abstract}

\newpage

\tableofcontents
 
\newpage

\section{Introduction}
\label{section:intro}

We present a general framework for constructing long locally correctable
codes from short base codes via the operation of lifting. Our notion of lifting
generalizes affine lifting, automorphic lifting, and high-degree sampling
defined in previous works, and we use it to obtain
new explicit high rate locally correctable codes by lifting
certain algebraic geometric codes.

\subsection{Error correcting codes and locally correctable codes}
We begin with some coding theory preliminaries.
A \emph{code} $\calC$ of block length $N$ over an alphabet $R$ is a subset of
$R^N$. Elements $f \in \calC$ are
\emph{codewords}.
Typically $\Sigma$ is used to denote the alphabet, but we use
$R$ because it is helpful to think of a codeword $f$ not as a
vector in $R^N$, but as a function $f : D \to R$ ($D$ for domain,
$R$ for range), where we identify $D$ with $[N] = \{1,\ldots,N\}$.
Typically one thinks of $\calC$ as the image of
some \emph{encoding map} $\textsf{Enc}: R_0^K \to R^N$ which injectively maps
$K$-symbol \emph{messages} over an alphabet $R_0$ to
$N$-symbol codewords (here $R_0$ may be different from $R$).
The \emph{rate} of the code $\calC$ is $K/N$, which measures the efficiency
of our encoding. We want to make $K/N$ as large as we can.
Another important parameter of a code is the minimum pairwise distance between
distinct codewords. The \emph{(Hamming) distance} between two words
$f,g \in R^N$ is the number of coordinates in which they differ,
i.e.
$$
\Delta(f,g) \triangleq \{i \in [N] \mid x_i \ne y_i \}.
$$
The distance $\Delta(\calC)$ of $\calC$ is simply
$\min \{\Delta(f,g) \mid f,g \in \calC,~f \ne g\}$. We want
$\Delta(\calC)$ to be as large as possible. We often look at
the normalized distance $\delta(f,g)$, which is simply
$\frac1N \Delta(f,g)$, and similarly $\delta(\calC) = \frac1N \Delta(\calC)$.

The motivation behind error correcting codes is to make information
robust to noise. The original message $m \in R_0^K$ is encoded into
some codeword $\textsf{Enc}(m) \in R^N$. Noise may corrupt some symbols of
$\textsf{Enc}(m)$, resulting in a new word $r \in R^N$, the \emph{received
word}. The number of symbols corrupted is exactly $\Delta(\textsf{Enc}(m),r)$.
If the number of errors is small, say less than
$\Delta(\calC)/2$, then $\textsf{Enc}(m)$ is the unique codeword
in $\calC$ within Hamming distance $\Delta(\calC)/2$ of $r$,
and one can uniquely decode $r$ to get $m$, since $\textsf{Enc}$ is
injective.

To decode a received word, it may be necessary to examine the entire
word. In some settings, the received word is prohibitively large, and
one wishes only to decode one symbol of the message. Codes with which one
can do this by querying only a small number of symbols of the input 
are known as \emph{locally decodable codes}. A related concept
is the notion of \emph{locally correctable code}. Such a code allows one to
correct a symbol of the \emph{codeword} (rather than a symbol of the message)
by querying only a few symbols of the input. The main parameters of interest
are the rate and the query complexity, or locality, the number of symbols
queried to recover a single symbol. These codes are the focus
of this work. We formally define these notions in Section~\ref{section:prelim}.

\subsection{Previous work}

Until recently, there were no known locally correctable codes with
sublinear query complexity and rate greater than $1/2$.
The Reed-Muller code was the first locally correctable code, with
the first correction procedure proposed by Reed~\cite{ReedM},
which happened to be a local correction procedure.
The $m$-variate Reed-Muller over $\F_q$ with degree parameter $r$
consists of all $m$-variate polynomials of total degree less than
$r$. More precisely, a codeword is the list of evaluations of such
a polynomial on all points of $\F_q^m$.
The idea behind the local correction procedure is to pick a random
line passing through the point whose value we wish to correct, view
the restriction of the polynomial to the line as a corrupted
Reed-Solomon codeword, and use a Reed-Solomon decoding algorithm to
correct the value on the line.
For any $\epsilon > 0$, the Reed-Muller codes can achieve
query complexity $N^\epsilon$ by taking $m = 1/\epsilon$ and
$N = q^m$. Unfortunately, the $m$-variate Reed-Muller code with
positive distance (by taking $r$ to be a constant fraction of $q$)
can never exceed $1/m!$ in rate. This certainly never exceeds $1/2$.

The recent work of Kopparty, Saraf, and Yekhanin~\cite{KSY} introduced
the first locally correctable codes that can achieve rate greater
than $1/2$, and in fact can achieve any rate arbitrarily close to~$1$.
More precisely, for any $\epsilon, \alpha > 0$, the multiplicity code
can achieve query complexity $N^\epsilon$ and rate $1-\alpha$
while correcting a constant fraction of errors. One may view multiplicity
codes as a variant of Reed-Muller codes, where each codeword
consists of evaluations of a low-degree polynomial along with its
partial derivatives.

An alternative to the multiplicity codes are the lifted Reed-Solomon codes
of Guo, Kopparty, and Sudan~\cite{GKS13}.
These are yet another variant of Reed-Muller codes --- more precisely, they
are supercodes of Reed-Muller codes with vastly greater dimension
but the same distance. The main idea behind lifted codes is the notion
of ``lifting'' --- an operation first introduced in~\cite{BMSS}
to prove negative results in property testing.
Essentially, the lifting operation takes a short base code
$\calC \subseteq \{\F_q^t \to \F_q\}$
and ``lifts'' it to a longer code $\calC' \subseteq \{\F_q^m \to \F_q\}$,
for $m > t$,
such that codewords of $\calC'$ are those $f:\F_q^m \to \F_q$
whose restriction to every $t$-dimension affine subspace is a codeword
of $\calC$. Guo et al~\cite{GKS13} obtain locally correctable codes
with query complexity $N^\epsilon$ and rate $1-\alpha$ by
lifting the Reed-Solomon code. Our work generalizes this
notion of lifting.

The work of Ben-Sasson et al~\cite{BGKKS13} presents another way to
build long locally correctable codes from short base codes via
the ``degree-lifting'' operation. Degree-lifting abstracts the process
of obtaining the Reed-Muller codes from the Reed-Solomon code and applies
it to algebraic geometry codes.
By degree-lifting certain algebraic geometry codes, such as
the Hermitian code, Ben-Sasson et al obtain locally correctable codes
with Reed-Muller-like properties but significantly smaller alphabet.
Unfortunately, degree-lifting faces the same rate barrier that
the Reed-Muller codes face, for essentially the same reason.
Two key contributions of~\cite{BGKKS13}
which we use in our work are the notions of a group being
``close'' to doubly transitive, and the fractal correction algorithm.
In particular, a conceptual contribution of~\cite{BGKKS13} is
the observation that the ``uniformity'' of the automorphism group
of an algebraic geometry code yields good local correctability
properties. Our work generalizes this observation.
Ben-Sasson et al also suggests the idea of
``automorphic lifting'', a natural generalization of the affine lifting
of~\cite{GKS13} to apply to algebraic geometry codes.
Our work further generalizes this idea.
Moreover, our notion of lifting encapsulates the notion of
high-degree sampling used in~\cite{BGKKS13} as well.
The idea of high-degree sampling is to restrict not to automorphisms,
but to ``high-degree views''. For instance, instead of restricting
to lines to decode the Reed-Muller code, one may restrict
to curves parametrized by quadratic functions.

\subsection{Our results}
In this work, we introduce a lifting framework which abstracts the lifting
operation used by~\cite{GKS13}
and the automorphic lifting suggested by~\cite{BGKKS13} as well
as the high-degree restrictions used by~\cite{BGKKS13}.
Our framework applies to \emph{arbitrary codes} and
\emph{arbitrary sets of functions} (as opposed to invariant codes
under some group of (generalized) automorphisms).
In particular, unlike the degree-lifting operation of~\cite{BGKKS13},
our lifting operation does not require an algebraic notion
of ``degree''.
Informally,
our lifting operation is defined as follows.
Let $\Phi$ be a set of functions from $D \to D$.
The $m$-variate lift of $\calC \subseteq \{D \to R\}$
with respect to $\Phi$ is the code whose codewords are those
$f : D^m \to R$ such that the univariate function
$f(\sigma_1(x),\ldots,\sigma_m(x))$ is a codeword of $\calC$
for all $(\sigma_1,\ldots,\sigma_m) \in \Phi^m$.
For affine-lifting, the domain is $D = \F_q$ and
$\Phi$ is the group of affine permutations on $\F_q$,
and in~\cite{GKS13} the base code is taken to be affine-invariant.
More generally, for automorphic lifting, $\Phi$ is some group
of automorphisms on $D$ under which $\calC$ is invariant.
Our definition of lifting requires neither $\calC$ to be
$\Phi$-invariant, nor even $\Phi$ to be a~group.

A conceptual contribution of our work is to show that if $\Phi$ is
sufficiently close to uniform in the sense of Ben-Sasson et al~\cite{BGKKS13},
then this suffices for the lift to have good distance and be
locally correctable. We show that
there is nothing essential about the symmetry of the base code
under $\Phi$, nor the fact that $\Phi$ is a group.
Thus, designing good lifted codes ``merely'' involves choosing a good set
$\Phi$ with respect to which to lift. On the one hand,
including too many functions in $\Phi$ kills the rate of the lifted
code, since every function adds a constraint on the lifted code.
On the other hand, including too few functions in $\Phi$ kills
the distance of the lifted code, since we want enough functions
in $\Phi$ to make it ``close'' to doubly transitive.

As an application, we construct an explicit family of locally
correctable codes via lifting. The family arises
from lifting the Hermitian code, the algebraic geometry code
that \cite{BGKKS13} degree-lift. We obtain high rate locally correctable
codes similar to the lifted Reed-Solomon codes, except
over a significantly smaller alphabet. 

Though our explicit construction uses algebraic geometry codes
as base codes, our exposition is elementary and self-contained.
Invoking the language of algebraic function field theory is
only necessary to prove the properties of the base codes; the properties
themselves can be stated in elementary terms, and we do so.
We refer the interested reader who wishes to see the proofs of
these facts to the book of Stichtenoth~\cite{Sti93} on
algebraic function fields and codes.

\subsection{Comparison of parameters}
\label{ssec:comparison}
We compare the parameters of the constant rate locally correctable codes
found in the literature, including the ones constructed in this paper.
We start with some easy comparisons. The lifted Reed-Solomon code
of Guo, Kopparty, Sudan~\cite{GKS13} is strictly better than the
Reed-Muller code, as it is a strict supercode with the same distance.
In fact, with $m$ variables over $\F_q$, the two codes have the same
length, alphabet, and query complexity, but the rate of
Reed-Muller is bounded above by $\frac{1}{m!}$ (even as its distance goes
to $0$) whereas the rate of the
lifted Reed-Solomon code approaches $1$ as its distance goes to $0$.
Similarly, the lifted Hermitian code (Theorem~\ref{theorem:main})
has the same length, alphabet, and query complexity as that of
the degree-lifted Hermitian code of Ben-Sasson et al~\cite{BGKKS13},
but the rate of the degree-lifted Hermitian code is bounded
above by $\frac{1}{m!}$ whereas the rate of the lifted
Hermitian code approaches $1$ as its distance goes to $0$.

To compare the various families of high rate locally correctable codes,
we normalize their parameters. Namely, we fix the block length to $N$,
the rate to $1-\alpha$, query complexity to $N^\epsilon$, and compare
the alphabet size and error correcting rate of each code. The results
are summarized in the table below.

\medskip
\begin{center}
\begin{tabular}{|l|r|r|}
\hline
Code & Alphabet size & Error correcting rate \\
\hline
Multiplicity~\cite{KSY} & $N^{\Omega((1/\epsilon)^{(1/\epsilon)})}$ & 
$\Omega(\epsilon \alpha)$ \\
\hline
Lifted Reed-Solomon~\cite{GKS13}	&	$N^\epsilon$	&
$\alpha^{O((2/\epsilon)^{(1/\epsilon)}\log(1/\epsilon))}$ \\
\hline
Lifted Hermitian (Theorem~\ref{theorem:main}) & $N^{\epsilon/3}$ &
$\alpha^{O((8/\epsilon)^{(2/\epsilon)}\log(1/\epsilon))}$ \\
\hline
\end{tabular}
\end{center}
\bigskip

In order for the lifted Reed-Solomon to match the alphabet size
of the lifted Hermitian code (by taking locality $N^{\epsilon/3}$),
its error correcting rate must become
$\alpha^{O((6/\epsilon)^{(3/\epsilon)}\log(1/\epsilon))}$
which is worse than that of the lifted Hermitian code
for sufficiently small $\epsilon$.

In comparison with the multiplicity codes of~\cite{KSY}, the lifted Hermitian
code achieves a much smaller alphabet but also much poorer (though
still positive constant) error correction
rate. The smaller alphabet is not necessarily an advantage, since one
can simply concatenate the multiplicity codes with a suitably good
linear code over an alphabet of constant size and still achieve
$N^\epsilon$ locality, $1-\alpha$ rate, and constant distance.
However, the lifted Hermitian code may outperform the multiplicity code
in certain concrete settings of parameters.

\paragraph{Organization.}
In Section~\ref{section:prelim} we introduce standard notation and terminology
used in the paper. In Section~\ref{section:def} we present the key definitions
and notions used in the paper, in particular the definition of lifting.
In Section~\ref{section:distance} we show that if a set
of functions is sufficiently ``close to doubly transitive'',
lifting a code with respect
to the set yields a code with good distance.
In Section~\ref{section:correction} we show in addition that
the lifted codes are locally correctable.
We emphasize that Sections~\ref{section:def},~\ref{section:distance},
and~\ref{section:correction} apply to arbitrary base codes, not necessarily
algebraic or even linear codes. In Section~\ref{section:basecodes},
we introduce the base codes used in our constructions. We review
the Reed-Solomon code as a warmup, and then present the Hermitian code
which we lift in Section~\ref{section:codes} to obtain
explicit high rate locally decodable codes with small alphabet size.
We conclude in Section~\ref{section:conclusion}.

\section{Preliminaries}
\label{section:prelim}

\subsection{Notation}
For integers $a < b$, let $[a,b]$ denote the the set
$\{a,a+1,a+2,\ldots,b\}$ and let $[a]$ denote $[1,a]$.
Throughout the paper, we let $\Phi$ denote a set of functions
mapping $D \to D$. We assume that $\Phi$ contains the identity
$\id : D \to D$ which fixes every element of $D$.
We say $\Phi$ \emph{acts} on $D$.

Let $f:D \to R$ and let $\sigma \in \Phi$ where $\Phi$ acts on $D$.
The function $f \circ \sigma: D \to R$ is defined by
$$
(f \circ \sigma)(x) = f(\sigma(x))
$$
for all $x \in D$.
Let $m \ge 1$ and let $\sigma = (\sigma_1,\ldots,\sigma_m) \in \Phi^m$.
For a function $f: D^m \to R$, define the function $f|_\sigma: D \to R$
by
$$
(f|_\sigma)(x) = f(\sigma_1(x),\ldots,\sigma_m(x))
$$
for all $x \in D$.
For a set $\Phi$ acting on $D$ and a point $u \in D^m$,
define the \emph{endomorphisms passing through $u$} to be
$$
\Phi_u \triangleq
\{ \sigma \in \Phi^m \mid \sigma_1 = \id,
\sigma_i(u_1) = u_i \,\forall i \in [2,m]\}.
$$
For an event $A$, let $\1_A$ denote the indicator variable for $A$,
i.e.\ 
$$
\1_A = \begin{cases}
1 & {\rm~if~} $A$ \\
0 & {\rm~otherwise}.
\end{cases}
$$
Let $f,g : D \to R$. The \emph{(relative) distance between $f$ and $g$},
is
$$
\delta(f,g) \triangleq \E_{x \in D} [\1_{f(x) \ne g(x)}].
$$
For a collection $\calC \subseteq \{D \to R\}$ of functions,
the distance between $f:D \to R$ and $\calC$ is
$$
\delta(f,\calC) \triangleq \min_{g \in \calC} \delta(f,g).
$$
For a code $\calC \subseteq \{D \to R\}$, the distance of $\calC$
is
$$
\delta(\calC) \triangleq \min_{\substack{f,g \in \calC \\ f\ne g}}
\delta(f,g)
$$
If $q$ is a prime power, let $\F_q$ denote the finite field of order $q$,
which is unique up to isomorphism.

\subsection{Terminology}
For an algorithm $\calA$ and function $f$, let $\calA^f$ denote
the algorithm $\calA$ given oracle access to $f$.

\begin{definition}[Locally correctable code]
A code $\calC \subseteq \{D \to R\}$ is
\emph{$(q,\tau)$-locally correctable} if there exists a randomized
algorithm $\calA$ satisfying the following properties:
\begin{enumerate}
\item%
$\calA^f$ makes at most $q$ queries to $f$;

\item%
If there exists $g \in \calC$ such that $\delta(f,g) \le \tau$, then
for every $x \in D$ we have $\calA^f(x) = g(x)$ with probability at least
$2/3$ over the randomness of $\calA$.
\end{enumerate}
\end{definition}

\begin{definition}[Locally decodable code]
A code $\calC \subseteq \{D \to R\}$ is
\emph{$(q,\tau)$-locally decodable} if $\calC$ is the image of an encoding
function $\textsf{Enc}: R^k \to R^D$ and there exists a randomized algorithm
$\calA$ satisfying the following properties:
\begin{enumerate}
\item%
$\calA^f$ makes at most $q$ queries to $f$;

\item%
If there exists $m \in R^{k}$ such that $\delta(f,\textsf{Enc}(m)) \le \tau$,
then for every $i \in [k]$ we have
$\calA^f(i) = m_i$ with probability at least $2/3$ over the
randomness of $\calA$.
\end{enumerate}
\end{definition}

For linear codes, local correctability is stronger than local decodability,
since one can arrange the generator matrix of the code such that the
message is part of the codeword.

\section{Definitions}
\label{section:def}
In this section we give the key definitions in the paper, namely
$\Phi$-lifting and the notion of a set $\Phi$ being
``close'' to doubly transitive, which is borrowed from~\cite{BGKKS13}.

\subsection{Lifting}

\begin{definition}
\label{def:lift}
Let $\Phi$ act on $D$ and
let $\calC \subseteq \{D \to R\}$.
The \emph{$m$-dimensional $\Phi$-lift} of $\calC$, denoted
$\lift_\Phi^m(\calC)$, is the set
\[
\lift_\Phi^m(\calC) \triangleq
\{f: D^m \to R \mid f|_\sigma \in \calC {\rm ~for~all~} \sigma \in \Phi^m\}.
\]
\end{definition}


We say $\calC \subseteq \{D \to R\}$ is \emph{$\Phi$-invariant} if
whenever $f \in \calC$ and $\sigma \in \Phi$ we also have $f \circ \sigma
\in \calC$. Notice that Definition~\ref{def:lift} does not require
that $\calC$ be $\Phi$-invariant, or even that $\Phi$ be a group!
Indeed, $\Phi$-invariance only ensures us that
$$
\lift_\Phi^1(\calC) = \calC
$$
and if in addition  is a group, then the lift operation composes:
$$
\lift_{\Phi^n}^m(\lift_\Phi^n(\calC)) = \lift_\Phi^{mn}(\calC)
$$
where $\Phi^n$ acts on $D^m$ componentwise, i.e.\
if $\varphi = (\varphi_1,\ldots\varphi_m) \in \Phi^m$ and
$x = (x_1,\ldots,x_m) \in D^m$, then
$\varphi(x) = (\varphi_1(x_1),\ldots,\varphi_n(x_n))$.

The affine lifting found in~\cite{GKS13} is (almost) an example of our notion
of lifting. Take $D = R = \F_q$ and $\Phi$ to be the group of affine
 permutations
on $D$, i.e. maps of the form $x \mapsto ax+b$ for $a \in \F_q^*$,
$b \in \F_q$. Then $\lift_\Phi^m(\calC)$ consists of all $f:\F_q^m \to \F_q$
such that $f|_L \in \calC$ for all lines $L$ that are not axis-parallel.
The affine-lifted codes in~\cite{GKS13} consider every line, including
the axis-parallel ones. Though we could have defined $\Phi$-lifting to properly
generalize affine-lifting, we chose our definition because it is cleaner
to state, makes proofs cleaner, and makes negligible difference in the
parameters we care about. We point out that one limitation of our definition
is that we can only lift a domain $D$ to a direct product $D^m$, whereas
the affine lifting of~\cite{GKS13} allows lifting from $\F_q^m$ to
$\F_q^n$ for any $m \le n$.

Though any code can be lifted, our constructions in the paper use
linear codes as the base code.
A code $\calC \subseteq \{D \to R\}$ is \emph{linear} if
$R = \F$ is a field and $\calC$ is a $\F$-vector space.
To argue that the lifted code is large,
we argue that it has large dimension by showing it contains many
linearly independent codewords. To do so, we need the following fact,
which is straightforward to verify.

\begin{proposition}
\label{prop:liftlinear}
If $\calC$ is linear over $\F$, then so is $\lift_\Phi^m(\calC)$.
\end{proposition}

\subsection{Double transitivity}

Now we define the notions of ``closeness'' to double transitivity that we
will work with. There are two such notions, taken from~\cite{BGKKS13}.

\begin{definition}
A set $\Phi$ acting on a set $D$ is \emph{doubly transitive} if it is
transitive on pairs in $\Phi$, i.e. for every $x_1 \ne x_2 \in D$ and
$y_1 \ne y_2 \in D$, there exists $\sigma \in \Phi$ such that
$\sigma(x_1) = y_1$ and $\sigma(x_2) = y_2$.
\end{definition}

\begin{definition}[\cite{BGKKS13}]
\label{def:ea}
A set $\Phi$ acting on a set $D$ is \emph{$(\epsilon,\alpha)$-doubly transitive}
if, for every $x_1,x_2 \in D$, for at least $1-\epsilon$ fraction of points
$x \in D$, the random variable $\sigma(x)$ is uniformly distributed on
$1-\alpha$ fraction of $D$, where $\sigma$ is chosen uniformly from the set
$\{\sigma \in \Phi \mid \sigma(x_1) = x_2\} = \Phi_{(x_1,x_2)}$.
\end{definition}

When $\Phi$ is a group acting transitively on $D$,
double transitivity is equivalent to
 $\left(\frac{1}{|D|},0\right)$-double
transitivity (see~\cite[Lemmas~6.8, 6.9]{BGKKS13}).
Indeed, given $x_1,x_2 \in D$, for every point $x \ne x_1$,
the random variable $\sigma(x)$ is uniformly distributed on $D$, when
$\sigma$ is drawn from those mapping $\sigma(x_1) = x_2$.
However, $\sigma(x_1)$ itself will always equal $x_2$.

\begin{example}
\label{example:RSdoubletrans}
Let $D = \F_q$ and $\Phi = \{x \mapsto ax+b \mid a \in \F_q^*, b \in \F_q\}$.
Then $\Phi$ is $(\frac1q,0)$-double transitive. This follows from the fact that
$\Phi$ is doubly transitive on $D$. Another way to see this is to note that,
given $x_1,x_2 \in D$, $\sigma(x_1) = x_2$ implies
$\sigma(x) = a(x-x_1) + x_2$ for some $a \in \F_q^*$. Therefore, for
every $x \ne x_1$ and every $y \in \F_q$, there exists a unique $\sigma$
such that $\sigma(x) = y$, namely the one with $a = (y-x_2)(x-x_1)^{-1}$.
\end{example}

The second notion of ``closeness'' to double transitive involves distributions
that are statistically close to uniform. The precise definition is as follows.

\begin{definition}
Let $p_1,p_2$ be two distributions on $D$, i.e. $\sum_{x \in D} p_1(x) = \sum_{x \in D} p_2(x) = 1$ and $p_1(x), p_2(x) \ge 0$ for all $x \in D$.
The \emph{distance between $p_1$ and $p_2$} is
$$
\|p_1 - p_2\| \triangleq \max_{A \subseteq D}
\left| \sum_{x \in A} p_1(x) - \sum_{x \in A} p_2(x) \right|.
$$
\end{definition}

\begin{definition}[\cite{BGKKS13}]
\label{def:csteps}
A set $\Phi$ acting on a set $D$ is \emph{$(\alpha,\epsilon)$-close to
$c$-steps uniform} if, for every $x_1,x_2 \in D$, for at least $1-\epsilon$
fraction of points $x \in D$, if one uniformly randomly chooses
$w_1,\ldots,w_{c-1} \in D$ and $\sigma_1,\ldots,\sigma_c \in \Phi$ such that
$\sigma_1(x_1) = x_2$ and $\sigma_i(w_{i-1}) = \sigma_{i-1}(w_{i-1})$
for $2 \le i \le c$, then the random variable $\sigma_c(x)$ is
$\alpha$-close to uniformly distributed on $D$.
\end{definition}

The motivation behind Definition~\ref{def:csteps} is the use of
\emph{fractal correcting} in~\cite{BGKKS13}. Intuitively, one may think
of $f|_\sigma$ as $f$ restricted to some curve in $D^m$. For simplicity
assume $m=2$. To correct the
received word $f$ at a particular point $x$, the usual approach is to
pick a random curve passing through $x$ and correct the shorter
word $f$ restricted to the curve. Parametrize the curve by $(x,\sigma(x))$.
Then the condition that the curve passes through $x = (x_1,x_2)$ is
equivalent to $\sigma(x_1) = x_2$.
If the curve samples $D$ uniformly, then with high
probability the curve does not contain too many corrupted points.
If $\Phi$ is not
doubly transitive, however, then random curves may not sample $D^2$
uniformly. The intuition behind fractal correcting is to first pick a random
curve $\sigma_1$ passing through $x$ (i.e. $\sigma_1(x_1) = x_2$),
then pick a random point $(w_1,\sigma_1(w_1))$ sitting on the point,
then pick another random curve $\sigma_2$ passing through $(w_1,\sigma_1(w_1))$
(i.e. $\sigma_2(w_1) = \sigma_1(w_1)$) and so on. After $c$ steps, the
$c$th curve $\sigma_c$ will sample the space nearly uniformly.
We elaborate on this in Sections~\ref{section:distance}
and~\ref{section:correction}.

\section{Distance of lifted codes}
\label{section:distance}

In this section we show that if $\calC$ is a code with constant
positive distance, and the set $\Phi$ acting on the domain $D$ is nearly doubly
transitive, then $\lift_\Phi^m(\calC)$ has constant positive distance.
Our lower bound on the distance of $\lift_\Phi^m(\calC)$ degrades as $m$
grows, but for our purposes $m$ is constant, so the distance of the lift
is constant as well. We emphasize that the results in this section apply
to \emph{any} code $\calC$, even non-linear codes. 

We begin by lower bounding the distance of the lift when the set 
is close to doubly transitive, in the sense of Definition~\ref{def:ea},
i.e.\ when $\Phi$ is $(\epsilon,\alpha)$-double transitive.

The following lemma
will be used in proving both Theorems~\ref{theorem:dist1}
and~\ref{theorem:decode1}.

\begin{lemma}
\label{lemma:avg1}
Let $\Phi$ acting on $D$ be $(\epsilon,\alpha)$-double transitive.
Let $m \ge 1$ and let $f,g \in \{D^m \to R\}$.
Fix $x \in D^m$.
Then
$$
\E_{\sigma \in \Phi_x}[\delta(f|_\sigma,g|_\sigma)]
\le \epsilon + \frac{\delta(f,g)}{(1-\alpha)^{m-1}}.
$$
\end{lemma}
\begin{proof}
Let $D' \subseteq D$ be the set of $z \in D$
such that $\sigma(z)$ is uniform over $1-\alpha$ fraction of $D$,
when $\sigma$ is chosen uniformly from $\Phi_x$, as in
Definition~\ref{def:ea}.
Note that $|D'| \ge (1-\epsilon) |D|$.
We have
\begin{eqnarray*}
\E_{\sigma \in \Phi_x} \E_{z \in D} [\1_{f|_\sigma(z)\ne g|_\sigma(z)}]
&=& \E_{z \in D} \E_{\sigma \in \Phi_x} [\1_{f|_\sigma(z)\ne g|_\sigma(z)}] \\
&=& \E_{z \notin D'} \E_{\sigma \in \Phi_x} [\1_{f|_\sigma(z)\ne g|_\sigma(z)}]
+ \E_{z \in D'} \E_{\sigma \in \Phi_x} \1_{f|_\sigma(z)\ne g|_\sigma(z)}] \\
&\le& \epsilon +
\E_{z \in D'} \E_{\sigma \in \Phi_x}[\1_{f|_\sigma(z)\ne g|_\sigma(z)}] \\
&\le& \epsilon + \frac{\delta(f,g)}{(1-\alpha)^{m-1}}
\end{eqnarray*}
where the final inequality follows from the fact that the last $m-1$
coordinates of
$\sigma(z)$ are uniform over $(1-\alpha)^{m-1}$ fraction of $D^{m-1}$
and in the worst case all the disparate points of $f$ and $g$
all lie in this subset.
\end{proof}

\begin{theorem}
\label{theorem:dist1}
Let $\calC \subseteq \{D \to R\}$ be a code with distance
$\delta$, and $\Phi$ acting on $D$
is $(\epsilon,\alpha)$-doubly transitive. Then
$\delta(\lift_\Phi^m(\calC)) \ge (1-\alpha)^{m-1}(\delta - \epsilon)$.
\end{theorem}
\begin{proof}
Let $f,g \in \lift_\Phi^m(\calC)$ be distinct and fix
$x \in D^m$ such that $f(x) \ne g(x)$.
By Lemma~\ref{lemma:avg1},
$$
\E_{\sigma \in \Phi_x}[\delta(f|_\sigma,g|_\sigma)]
\le \epsilon + \frac{\delta(f,g)}{(1-\alpha)^{m-1}}.
$$
Therefore, there exists $\sigma \in \Phi_x$ such that
$\delta(f|_\sigma,g|_\sigma) \le
\epsilon + \frac{\delta(f,g)}{(1-\alpha)^{m-1}}$.
But $f|_\sigma(x_1) = f(x) \ne g(x) = g|_\sigma(x_1)$,
so $f|_\sigma$ and $g|_\sigma$ are distinct codewords of $\calC$ and hence
$\delta \le \epsilon + \frac{\delta(f,g)}{(1-\alpha)^{m-1}}$, i.e.
$\delta(f,g) \ge (1-\alpha)^{m-1}(\delta - \epsilon)$.
\end{proof}

Next we prove a similar result when $\Phi$ is close to doubly transitive
in the sense of Definition~\ref{def:csteps}, i.e.\  is
to $(\alpha,\epsilon)$-close to $c$-steps uniform.
First, some straightforward but useful facts.

\begin{lemma}
If $X$ and $Y$ are independent and $X$ is $\alpha$-close to uniform over
$S$ and $Y$ is $\beta$-close to uniform over $T$, then
$(X,Y)$ is $\alpha+\beta$-uniform over $S \times T$.
\end{lemma}

\begin{corollary}
If $X_i \in D$ is $\alpha$-close to uniform for each $i \in [m]$ and are
independent, then
$(X_1,\ldots,X_m) \in D^m$ is $m \alpha$-close to uniform.
\end{corollary}

The following lemma will be used in proving both
Theorems~\ref{theorem:dist2} and~\ref{theorem:decode2}.

\begin{lemma}
\label{lemma:avg2}
Let $\Phi$ acting on $D$ be $(\alpha,\epsilon)$-close to $c$-steps uniform.
Let $m \ge 1$ and let $f,g \in \{D^m \to R\}$. Fix
$x \in D^m$. Then
$$
\E_{\sigma_1 \in \Phi_x} \E_{w_1 \in D} \E_{\sigma_2 \in \Phi_{\sigma_1(w_1)}}
\cdots \E_{\sigma_c \in \Phi_{\sigma_{c-1}(w_{c-1})}}[\delta(f|_{\sigma_c},
g|_{\sigma_c})] \le \delta(f,g) + \epsilon + m\alpha.
$$
\end{lemma}
\begin{proof}
Let $D' \subseteq D$ be the set of $z \in D$ such that
$\sigma_c(z)$ is $\alpha$-close to uniform, as in Definition~\ref{def:csteps}.
Note that $|D'| \ge (1-\epsilon) |D|$.
Then
\begin{eqnarray*}
&& \E_{\sigma_1 \in \Phi_x} \E_{w_1 \in D} \E_{\sigma_2 \in
 \Phi_{\sigma_1(w_1)}}
\cdots \E_{\sigma_c \in \Phi_{\sigma_{c-1}(w_{c-1})}}
\E_{z \in D} \left[ \1_{f|_{\sigma_c}(z)\ne g|_{\sigma_c}(z)} \right] \\
&=& \E_{z \in D}
\E_{\sigma_1 \in \Phi_x} \E_{w_1 \in D} \E_{\sigma_2 \in \Phi_{\sigma_1(w_1)}}
\cdots \E_{\sigma_c \in \Phi_{\sigma_{c-1}(w_{c-1})}}
\left[ \1_{f|_{\sigma_c}(z)\ne g|_{\sigma_c}(z)} \right] \\
&\le& \epsilon + \E_{z \in D'}
\E_{\sigma_1 \in \Phi_x} \E_{w_1 \in D} \E_{\sigma_2 \in \Phi_{\sigma_1(w_1)}}
\cdots \E_{\sigma_c \in \Phi_{\sigma_{c-1}(w_{c-1})}}
\left[ \1_{f|_{\sigma_c}(z)\ne g|_{\sigma_c}(z)} \right] \\
&\le& \epsilon + \delta(f,g) + m\alpha.
\end{eqnarray*}
\end{proof}

\begin{theorem}
\label{theorem:dist2}
Let $\calC$ be a code with distance $\delta$,
and $\Phi$ is $(\alpha,\epsilon)$-close to $c$-steps uniform.
Then $\delta(\lift_\Phi^m(\calC)) \ge \delta^c - m\alpha - \epsilon$.
\end{theorem}
\begin{proof}
Let $f,g \in \lift_\Phi^m(\calC)$ be distinct and let $\tau = \delta(f,g)$.
Fix $x \in D$ such that $f(x) \ne g(x)$.
We claim that, for each $i \in [c]$, there exists
 $w_{i-1} \in D$ and $\sigma_i \in \Phi_{\sigma_{i-1}(w_{i-1})}$ such that
$$
0 < \E_{w_i \in D} \E_{\sigma_{i+1} \in \Phi_{\sigma_i(w_i)}}
\cdots \E_{\sigma_c \in \Phi_{\sigma_{c-1}(w_{c-1})}}
\E_{z \in D} \left[ \1_{f|_{\sigma_c}(z)\ne g|_{\sigma_c}(z)} \right]
\le \frac{\tau+m\alpha + \epsilon}{\delta^{i-1}}.
$$
We prove the claim by induction. The base case $i=1$ follows
by taking $\sigma_0 \in \Phi_x$, $w_0 = x_1$, and noting that,
by Lemma~\ref{lemma:avg2}, since
$$
\E_{\sigma_1 \in \Phi_x} \E_{w_1 \in D} \E_{\sigma_2 \in \Phi_{\sigma_1(w_1)}}
\cdots \E_{\sigma_c \in \Phi_{\sigma_{c-1}(w_{c-1})}}
\E_{z \in D} \left[ \1_{f|_{\sigma_c}(z)\ne g|_{\sigma_c}(z)} \right]
\le \tau + m\alpha + \epsilon,
$$
there exists $\sigma_1 \in \Phi_x$ such that
$$
 \E_{w_1 \in D} \E_{\sigma_2 \in \Phi_{\sigma_1(w_1)}}
\cdots \E_{\sigma_c \in \Phi_{\sigma_{c-1}(w_{c-1})}}
\E_{z \in D} \left[ \1_{f|_{\sigma_c}(z)\ne g|_{\sigma_c}(z)} \right]
\le \tau + m\alpha + \epsilon.
$$
Moreover, this expectation is positive because $f(x) \ne g(x)$.
Now suppose we have proved the $i-1$ case. The restrictions
$f|_{\sigma_{i-1}}$ and $g|_{\sigma_{i-1}}$ are distinct codewords of $\calC$
(since
they disagree at $w_{i-2}$) and hence for at least $\delta$-fraction
of $w_{i-1} \in D$ we have
$f(\sigma_{i-1}(w_{i-1})) \ne g(\sigma_{i-1}(w_{i-1}))$.
Restricting to these $w_{i-1}$, we get
$$
0 < \delta \cdot \E_{\sigma_i \in \Phi_{\sigma_{i-1}(w_{i-1})}}
\E_{w_i \in D} \E_{\sigma_{i+1} \in \Phi_{\sigma_i(w_i)}}
\cdots \E_{\sigma_c \in \Phi_{\sigma_{c-1}(w_{c-1})}}
\E_{z \in D} \left[ \1_{f|_{\sigma_c}(z)\ne g|_{\sigma_c}(z)} \right]
\le \frac{\tau + m\alpha + \epsilon}{\delta^{i-2}}
$$
and the claim thus follows.

From the $i=c$ case of the claim, it follows that there exists
$\sigma_c \in \Phi$ such that
$$
0 < \E_{z \in D} \left[\1_{f|_{\sigma_c}(z)\ne g|_{\sigma_c}(z)} \right]
\le \frac{\tau+ m\alpha + \epsilon}{\delta^{c-1}}.
$$
Thus $f|_\sigma$ and $g|_\sigma$ are distinct codewords of $\calC$, so we have
$\delta \le \frac{\tau + m\alpha + \epsilon}{\delta^{c-1}}$.
\end{proof}

\section{Correction algorithms}
\label{section:correction}

In this section we describe how to locally correct a lifted code,
given a decoding algorithm for the base code. We present two correcting
methods. The first is one-shot correcting, which abstracts the local correcting
algorithms for Reed-Muller codes and the affine-lifted Reed-Solomon codes
of~\cite{GKS13}, and is also used for correcting degree-lifted AG codes
in~\cite{BGKKS13}. The idea is to pick a random curve passing through the point
which we would like to correct, view the restriction of the received word
to the curve as a received word that should be close to a codeword of the
base code, and then use the base code decoder to correct the point.
The second method is fractal correcting, which was introduced by Ben-Sasson
et al~\cite{BGKKS13}. The idea is to recursively perform one-shot correcting.
To correct a point, pick a random curve passing through it. However, now
recursively correct each point on the curve. If $\Phi$ is close to 
$c$-steps uniform,
then fractal correcting with recursion depth
$c$ should succeed with high probability. The analysis of the fractal
correction algorithm is found in~\cite{BGKKS13}, but we include a proof
here for completeness. We emphasize that, as in
Section~\ref{section:distance}, the results of this section apply to
arbitrary codes $\calC$.

\subsection{One-shot correcting}
The one-shot correcting algorithm $\calA$ works as follows.
To compute $\calA^f(x)$:
\begin{enumerate}
\item%
Pick $\sigma \in \Phi_x$ uniformly at random.
\item%
Use the decoding algorithm for $\calC$ to correct $f|_\sigma$ to
some function $g \in \calC$.
\item%
Output $g(x_1)$.
\end{enumerate}

\begin{theorem}
\label{theorem:decode1}
Let $\calC \subseteq \{D \to R\}$ be a code with distance
$\delta$ and suppose $\Phi$ is $(\epsilon,\alpha)$-doubly transitive.
Let $\calL = \lift_\Phi^m(\calC)$.
Suppose
$$
\delta(f,\calL) < (1-\alpha)^{m-1}\cdot
\min\{\delta/6 - \epsilon, (\delta-\epsilon)/2\}.
$$
Then there exists a unique
$\widehat f \in \calL$
such that $\delta(f,\widehat f) \le \delta(f,\calL)$ and for any $x \in D^m$
we have $\calA^f(x) = \widehat f(x)$ with probability at least $2/3$ over the
randomness of $\calA$.
\end{theorem}
\begin{proof}
By Theorem~\ref{theorem:dist1},
$\delta(\calL) \ge (1-\alpha)^{m-1}(\delta-\epsilon)$.
Since $\delta(f,\widehat f) < \delta(\calL)/2$, $\widehat f$ is unique.
Fix $x \in D^m$. By Lemma~\ref{lemma:avg1},
$$
\E_{\sigma \in \Phi_x}[\delta(f|_\sigma,\widehat f|_\sigma)] \le
\epsilon + \frac{\delta(f,\widehat f)}{(1-\alpha)^{m-1}} 
\le \epsilon + \frac{\delta(f,\calL)}{(1-\alpha)^{m-1}}.
$$
By Markov's inequality, with probability at least $2/3$,
$\delta(f|_\sigma,\widehat f|_\sigma) \le 3\left(\epsilon +
\frac{\delta(f,\calL)}{(1-\alpha)^{m-1}}\right) < \delta/2$.
Step~2 of the algorithm finds some $g \in \calC$ such that
$\delta(f|_\sigma,g) < \delta/2$. But both $g, \widehat f|_\sigma \in \calC$
and $\delta(g,\widehat f|_\sigma) < \delta$, so in fact
$g = \widehat f|_\sigma$. Therefore, $\calA^f(x) = g(x_1)
= \widehat f|_\sigma(x_1) = \widehat f(x)$.
\end{proof}

\begin{corollary}
\label{cor:decode1}
If $\calC \subseteq \{D \to R\}$ has distance $\delta$ and
$\Phi$ acting on $D$ is $(\epsilon,\alpha)$-doubly transitive, then
$\lift_\Phi^m(\calC)$ is $(q,\tau)$-locally correctable for $q = |D|$ and
$\tau = O((1-\alpha)^{m-1}(\delta-\epsilon))$.
\end{corollary}

\subsection{Fractal correcting}

The $c$-step fractal correction algorithm $\calA_c$ works as follows.
To compute $\calA_c^f(x)$:

\begin{enumerate}
\item%
If $c=1$, output $\calA^f(x)$.

\item%
Otherwise, $c > 1$. Pick $\sigma \in \Phi_x$ uniformly at random.

\item%
Compute $f' \triangleq \calA^f_{c-1}|_\sigma$.
That is, for each $z \in D$ let $f'(z) = \calA^f_{c-1}(\sigma(z))$.

\item%
Use the decoding algorithm for $\calC$ to correct $f'$ to some function
$g \in \calC$.

\item%
Output $g(x_1)$.

\end{enumerate}

\begin{theorem}
\label{theorem:decode2}
Let $\calC \subseteq \{D \to R\}$ be a code with distance
$\delta$ and suppose $\Phi$ acting on $D$
is $(\alpha,\epsilon)$-close to $c$-steps uniform. Let $\calL = \lift_\Phi^m(\calC)$. Suppose
$$
\delta(f,\calL) < \min\left\{\frac13 (\delta/2)^c - \epsilon - m\alpha,
(\delta^c - \epsilon - m\alpha)/2\right\} .
$$
Then there exists a unique $\widehat f \in \calL$ such that
$\delta(f,\widehat f) \le \delta(f,\calL)$ and for any $x \in D^m$
we have $\calA^f_c(x) = \widehat f(x)$ with probability at least
$2/3$ over the randomness of $\calA$.
\end{theorem}
\begin{proof}
By Theorem~\ref{theorem:dist2},
$\delta(\calL) \ge \delta^c - \epsilon - m\alpha$.
Since $\delta(f,\widehat f) < \delta(\calL)/2$, $\widehat f$ is unique.
Fix $x \in D^m$. For $i \in [c]$, let $p_i$ denote the average probability
that the $i$th bottom-most level of the recursion fails.
Our goal is to show that $p_c \le 1/3$. We will show in fact that
$p_i \le \frac13(\delta/2)^{c-i}$ for all $i \in [c]$.
By Lemma~\ref{lemma:avg2}, the average of $\delta(f|_{\sigma_c},\widehat f|_{\sigma_c})$ over all $\sigma_c$
chosen in the bottom-most level is at most
$\delta(f,\calL) + \epsilon + m\alpha$, so by Markov's inequality
with probability at most $\frac2\delta (\delta(f,\calL) + \epsilon + m\alpha)$
we have $\delta(f|_{\sigma_c},\widehat f|_{\sigma_c}) > \delta/2$, i.e.
$p_1 \le \frac2\delta(\delta(f,\calL) + \epsilon + m\alpha) \le\frac13
(\delta/2)^{c-1}$.

Now inductively assume $p_i \le \frac13(\delta/2)^{c-i}$.
The average value of $\delta(f|_{\sigma_{c-i+1}},\widehat f|_{\sigma_{c-i+1}})$
is at most $p_i$. By Markov's inequality,
with probability at most $\frac2\delta p_i$ we have
$\delta(f|_{\sigma_{c-i}},\widehat f|_{\sigma_{c-i}}) > \delta/2$, so
$p_{i+1} \le \frac2\delta p_i \le \frac13 (\delta/2)^{c-(i+1)}$.
\end{proof}

\begin{corollary}
\label{cor:decode2}
If $\calC \subseteq \{D \to R\}$ has distance $\delta$
for some $\Phi$ that is $(\alpha,\epsilon)$-close to $c$-steps uniform, where
$c = O(1)$, then $\lift_\Phi^m(\calC)$ is $(q,\tau)$-locally correctable for
$q = |D|^c$ and $\tau = O(\delta^c - \epsilon - m\alpha)$.
\end{corollary}

\section{Base codes}
\label{section:basecodes}
In this section we review existing codes, in particular the Reed-Solomon
code and the Hermitian code, the latter which
we use in Section~\ref{section:codes} to construct new high rate locally
correctable codes over small alphabets.

\paragraph{Algebraic geometry codes.} The Reed-Solomon and Hermitian
codes are instances of \emph{algebraic geometry codes}. Since we
can describe our base codes, our lifted codes, and their properties
without using any terminology typically used in the context of AG codes
(e.g.\ the language of algebraic function fields), we avoid using such
terminology and stick to an elementary exposition. In fact, the only
deep results from the theory of algebraic function fields that we use
can be stated in elementary terms. The interested reader is referred
to~\cite{Sti93} for details on the theory of algebraic function fields
and codes.

\subsection{Reed-Solomon code}
\label{ssec:RS}
Let $q$ be a prime power.
The \emph{Reed-Solomon code} $\RS_q[r] \subseteq \F_q[x]/(x^q-x)$ can be defined as
$$
\RS[r] \triangleq
\span_{\F_q} \{x^i \mid i < r\}.
$$
It is a $[q,r,q-r+1]_q$-code. Note that its alphabet size $q = N$ where
$N$ is its block size. One can identify $\F_q[x]/(x^q-x)$ with
$\{\F_q \to \F_q\}$. Consider the group $\Phi$ consisting of all affine
permutations on $\F_q$, i.e.\ $\Phi = \{x \mapsto ax+b \mid a \in \F_q^*,
b \in \F_q\}$, which acts on $\F_q$. Clearly $\RS_q[r]$ is $\Phi$-invariant.
Moreover, $\Phi$ is doubly transitive (Example~\ref{example:RSdoubletrans})
and $|\Phi| = q(q-1)$, so it is just large enough to be doubly transitive.
In~\cite{GKS13}, it was shown that $\lift_\Phi^m(\RS_q[(1-\delta)q])$
has block length $q^m$, distance at least $\delta - \frac1q$ (which also
follows from Theorem~\ref{theorem:dist1}),
and rate at least $1 - \delta^{\Omega\left(\frac{1}{m^m\log m}\right)}$
when $q$ is a power of $2$.

\subsection{Hermitian code}
\label{ssec:herm}
Let $q$ be a prime power.
The \emph{Hermitian curve} $H \subseteq \F_{q^2}^2$ is the set
$$
H \triangleq \{(x,y) \mid N(x) = Tr(y)\}
$$
where $N:\F_{q^2} \to \F_q$ is the \emph{norm} $N(x) = x^{1+q}$ and
$Tr:\F_{q^2} \to \F_q$ is the \emph{trace} $Tr(x) = x + x^q$. It can be shown
that $N$ is multiplicative and is a surjective group homomorphism
from $\F_{q^2}^* \to \F_q^*$ (and hence a $(q+1)$-to-$1$ map on $\F_{q^2}^*$)
and that $Tr$ is additive and is a surjective $\F_q$-linear map from
$\F_{q^2} \to \F_q$ (and hence a $q$-to-$1$ map on $\F_{q^2}$).
It follows that $|H| = q^3$, since for every $x \in \F_{q^2}$ there are exactly
$q$ values of $y \in \F_{q^2}$ such that $Tr(y) = N(x)$.

The \emph{Hermitian code} $\Herm_q[r] \subseteq
\F_{q^2}[x]/(x^{q^2}-x, y^{q^2}-y, N(x)-Tr(y))$ is defined as
$$
\Herm_q[r] \triangleq
\span_{\F_{q^2}} \{ x^iy^j \mid qi + (q+1)j < r, j < q\}.
$$
It follows from the Riemann-Roch theorem that $\Herm_q[r]$ is a
$[q^3,r-g,q^3-r+1]_{q^2}$-code, where $g = \frac{q(q-1)}{2}$ is the
genus of the curve $H$ (one can also deduce this by counting the number
of ``degrees'' $d$ which cannot be obtained by a sum $qi + (q+1)j$).
Though the Hermitian code has a worse rate-distance
trade-off than the Reed-Solomon code, its alphabet size is significantly 
smaller ($q^2$ compared to a block length of $q^3$).

Consider the group $\Phi$ of maps $(x,y) \mapsto (ax+b, a^{q+1}y + ab^qx + c)$
for $a \in \F_{q^2}^*$, $(b,c) \in H$. One can verify that this a group
of order $q^3(q^2-1)$ acting on $H$ and moreover $\Herm_q[r]$ is
 $\Phi$-invariant.
For interesting values of $r$, the group $\Phi$ is the largest group under
which the Hermitian code is invariant~\cite{Xin95}.
The group $\Phi$ is not doubly transitive, but it is shown in~\cite{BGKKS13}
that it is almost doubly transitive, in both the senses of
Definitions~\ref{def:ea} and~\ref{def:csteps}. We recall the precise
statements.

\begin{proposition}[{\cite[Theorem~6.3]{BGKKS13}}]
\label{prop:hermclose1}
Let $\Phi$ be as above. Then $\Phi$ is $(\epsilon,\alpha)$-doubly transitive
for $\epsilon = \frac{1}{q^2}$ and $\alpha = 1 - \frac1q$.
\end{proposition}

\begin{proposition}[{\cite[Theorem~7.3]{BGKKS13}}]
\label{prop:hermclose2}
Let $\Phi$ be as above. Then $\Phi$ is $(\alpha,\epsilon)$-close
to $2$-steps uniform for $\alpha = \epsilon = \frac1q$.
\end{proposition}

Applying Theorem~\ref{theorem:dist2} and Corollary~\ref{cor:decode2} to
the above facts, we immediately get the following.

\begin{theorem}
\label{theorem:hermliftdecode}
Let $\Phi$ be the group of automorphisms on $H$ of the form
$(x,y) \mapsto (ax+b, a^{q+1}y + ab^qx + c)$. Let $r = (1-\delta)q^3$,
so that $\Herm_q[r]$ has distance $\delta$. Then
$\lift_\Phi^m(\Herm_q[r])$ has distance at least $\delta^2 - \frac{m}{q}$
and is $(q^6, O(\delta^2 - \frac{m}{q}))$-locally correctable.
\end{theorem}

Note that, though the $\Phi$-lift of $\Herm_q[(1-\delta)q^3]$ has
distance roughly $\delta^2$ which is less than that of the degree-lift,
whose distance is $\delta$ (see~\cite{BGKKS13}), its error correcting
capability is the same.

\section{Explicit Constructions}
\label{section:codes}



In this section we prove the following.

\begin{theorem}
\label{theorem:main}
Given $\epsilon, \alpha, N_0 > 0$, for infinitely many $N \ge N_0$
there exists a code of length $N$, rate $1-\alpha$, alphabet size
$N^{\epsilon/3}$ and is
 $(N^\epsilon,\alpha^{O((8/\epsilon)^{(2/\epsilon)}\log(1/\epsilon))})$-locally
  correctable.
\end{theorem}

We prove this using lifted Hermitian codes. We defer the proof to the end
of the section.

Let $m \ge 1$, let $q = 2^\ell > m$, let $c > 0$ such that
$\ell-c > \ceil{\log_2 m}$, and let $r = (1-2^{-c})q^3$. Let
$\Phi$ be the group of automorphisms on the Hermitian curve
$H \subseteq \F_{q^2}^2$ of the form
$(x,y) \mapsto (ax + b, a^{q+1}y + ab^qx + c)$, and let
$\calL = \lift_\Phi^m(\Herm_q[r])$.
By Theorem~\ref{theorem:hermliftdecode}, $\calL$
has distance $2^{-2c} - \frac{m}{q}$ and is
$O(q^6, O(2^{-2c} - \frac{m}{q}))$-locally
correctable. Its length is $q^{3m}$ and alphabet size is $q^2$. The only
missing parameter is the rate, to which we devote the rest of this
section.

After lifting, the domain of our code is
$$
H^m = \{(x_1,y_1,\ldots,x_m,y_m) \in \F_{q^2}^m \mid
N(x_k) = Tr(y_k) ~~\forall k \in [m]\}.
$$
A \emph{monomial on $H^m$} is a monomial of the form
$\prod_{k=1}^m x_k^{i_k}y_k^{j_k}$ with $i_k < q^2$ and $j_k < q$
for all $k \in [m]$. The reason for these conditions is to
ensure the monomials define distinct functions on $H^m$. In fact,
one can show the monomials on $H^m$ form a basis of $\{H^m \to \F_{q^2}\}$ as
a $\F_{q^2}$-vector space.

\begin{definition}
Let $p$ be a prime.
Let $a,b \in \N$ and consider their base $p$ representations
$a = \sum_{i \ge 0} a_ip^i$ and $b = \sum_{i \ge 0} b_ip^i$
where each $a_i,b_i \in [0,p-1]$.
Then \emph{$a$ is the in the $p$-shadow of $b$}, denoted
$a \le_p b$, if $a_i \le b_i$ for all $i$.
Moreover, for $a,b,c \in \N$, we say $(a,b) \le_p c$ if
$a_i + b_i \le c_i$ for all~$i$.
\end{definition}

The following generalized theorem of Lucas will be crucial for our
analysis later. For $a+b \le c$, we let ${a \choose b,c}$ denote the
standard trinomial coefficient $\frac{a!}{b!c!(a-b)!}$ which is the
coefficient of $x^by^c$ in the expansion of $(x+y+1)^a$. Note that
the standard binomial coefficient is ${a \choose b} = {a \choose b,0}$

\begin{theorem}[(Generalized) Lucas' theorem]
Let $a,b,c \in \N$ with $p$-ary representations given
by $a_i,b_i,c_i$.
Then
$$
{a \choose b,c} \equiv \prod_{i \ge 0} {a_i \choose b_i,c_i} \bmod p.
$$
In particular, ${a \choose b,c} \bmod p$ is nonzero only if $(b,c) \le_p a$.
\end{theorem}

Our strategy for lower bounding $\dim_{\F_{q^2}} \calL$ is to
lower bound the number of monomials on $H^m$ in $\calL$. For a monomial
$f(x_1,y_1,\ldots,x_m,y_m) = \prod_{k=1}^m x_k^{i_k} y_k^{j_k}$ and
a map $\sigma \in \Phi^m$ where
$\sigma_k(x,y) = (a_k x+b_k, a_k^{q+1}y + a_kb_k^qx + c_k)$, we have
\begin{eqnarray*}
f(\sigma(x,y)) &=&
\prod_{k=1}^m \left( a_kx + b_k\right)^{i_k}
\left( a_k^{q+1}y + b_k^q x + c_k \right)^{j_k} \\
&=& \prod_{k=1}^m
\left(
\sum_{d_k \le_p i_k} (\cdots) x^{d_k}
\right)
\left(
\sum_{(d'_k,e_k) \le_p j_k} (\cdots)x^{d'_k}y^{e_k}
\right) \\
&=&\sum_{\forall k \ d_k \le_p i_k, (d'_k,e_k) \le_p j_k}
(\cdots) x^{\sum_{k=1}^m d_k + d'_k} y^{\sum_{k=1}^m e_k}
\end{eqnarray*}
where the $(\cdots)$ indicate constants which do not matter. Thus,
the monomial $f$ is in $\calL$ if the following holds:
for all $k \in [m]$, for all $d_k \le_p i_k$ and all $(d'_k,e_k) \le_p j_k$,
after reducing the monomial
$x^{\sum_{k=1}^m d_k + d'_k} y^{\sum_{k=1}^m e_k}$ modulo
the ideal $I \triangleq (x^{q^2}-x, y^{q^2}-y, x^{q+1} - y^q - y)$, the
resulting sum of monomials $x^iy^j$ all satisfy $qi + (q+1)j < r$.
The basis of monomials on $H$ given by $x^iy^j$ with $i < q^2$ and $j < q$
provides a canonical way to reduce monomials modulo $I$. To reduce $x^iy^j$,
we perform the following steps. While $i \ge q^2$ or $j \ge q$, if
$i \ge q^2$, reduce $x^iy^j$ to $x^{i-q^2+1}y^j$; if
$j \ge q$, reduce $x^iy^j$ to $x^{i+q+1}y^{j-q} - x^iy^{j-q+1}$.
At each step, either the degree of $x$ is strictly decreasing or the degree
of $y$ is strictly decreasing, and the degree of $y$ never increases, so
this process will eventually terminate.

\begin{lemma}
\label{lemma:goodcond}
For $a \in \N$, let $a_i$ denote the $i$th bit in the binary representation
of $a$, i.e.\ $a = \sum_{i \ge 0} a_i 2^i$.
Let $b = 2 + \ceil{\log_2 m}$. Let
$$
{\rm Good} = \{(i_1,\ldots,i_m,j_1,\ldots,j_m)
\mid \exists s \in [2\ell-c,2\ell-b-1] \,\forall t \in [0,b]\, \forall k\,
(i_k)_{s+t} = (j_k)_{s+t} = 0\}.
$$
If $(i_1,\ldots,i_m,j_1,\ldots,j_m) \in {\rm Good}$, then
$\prod_{k=1}^m x_k^{i_k} y_k^{j_k} \in \lift_\Phi^m(\calC)$.
\end{lemma}
\begin{proof}
For $a \in \N$, the condition $a < r = (1-2^{-c})q^3$ is equivalent to
the condition $\exists s' \in [3\ell-c, 3\ell-1]$ such that $a_{s'} = 0$.
For each $k \in [m]$, fix $d_k \le_p i_k$ and $(d'_k,e_k) \le_p j_k$.
The hypothesis implies that $(d_k)_{s+t} = (d'_k)_{s+t} = (e_k)_{s+t} = 0$
for all $t \in [0,b]$. It suffices to show that after
reducing $x^{\sum_{k=1}^m d_k + d'_k} y^{\sum_{k=1}^m e_k}$ modulo $I$
into a sum of monomials $x^iy^j$ with $i < q^2$ and $j < q$, each of them
satisfies
$(i)_{t} = (j)_{t} = 0$ for some $t \in [2\ell-c,2\ell-1]$,
for this would imply
$$
(qi + (q+1)j)_{t+\ell} = (qi)_{t+\ell} + ((q+1)j)_{t+\ell} 
= (i)_{t} + (j)_{t} = 0
$$
and since $t+\ell \in [3\ell-c, 3\ell-1]$ this implies the lemma.

Let $d = \sum_{k=1}^m d_k + d'_k$ and let $e = \sum_{k=1}^m e_k$.
Consider three cases.

\paragraph{Case 1.} $d < q^2$, $e < q$. In this case, the monomoial
$x^dy^e$ does not reduce, so it suffices to show that
$(d)_{s+b} = (e)_{s+b} = 0$. The only way one of these is $1$ is
by carrying from the lower order bits, so we may ignore the higher
order bits and assume without loss of generality that
$(d_k)_{s'} = (d'_k)_{s'} = (e_k)_{s'} = 0$ for $s' \ge s+b$.
Then $d_k, d'_k, e_k < 2^s$, so
$\sum_{k=1}^m d_k + d'_k < (m2^{s+1}) < 2^{s+b}$
and thus $(d)_{s+b} = 0$ and similarly $\sum_{k=1}^m e_k < m2^s
< 2^{s+b}$ so $(e)_{s+b} = 0$.

\paragraph{Case 2.} $d \ge q^2$, $e < q$. In this case, the monomial
$x^d y^e$ reduces to $x^{d \bmod (q^2-1)} y^e$. By the previous case,
$(e)_{s+b} = 0$, so it only remains to show $(d \bmod (q^2-1))_{s+b} = 0$.
Doubling $d$ cyclically permutes the bits of $d \bmod (q^2-1)$. In particular,
$(2d \bmod (q^2-1))_i = (d \bmod (q^2-1))_{i-1 \bmod 3\ell-1}$. Then
$(2^{3\ell-1-s-b}d \bmod (q^2-1))_i = (d \bmod (q^2-1))_{i+s+b+1-3\ell}$.
Therefore, it suffices to show that
$(2^{3\ell-1-s-b} d \bmod (q^2-1))_{3\ell-1} = 0$. Since the bits of order
$[s,s+b]$ of $d_k,d'_k$ are zero, the bits of order
$[3\ell-1-b,3\ell-1]$ of $2^{3\ell-1-s-b}$ times $d_k,d'_k,e_k$ are zero,
hence $2^{3\ell-1-s-b} d_k \bmod (q^2-1) < 2^{3\ell-1-b}$ and similarly for
$d'_k$. Therefore
$\sum_{k=1}^m 2^{3\ell-1-s-b}(d_k + d'_k) \bmod (q^2-1) 
< 2^{3\ell-1}$
so
 $(2^{3\ell-1-s-b}d \bmod (q^2-1))_{3\ell-1} = 0$, and
in fact $(2^{3\ell-1-s-b}d \bmod (q^2-1))_{3\ell-2} = 0$,
so we can conclude that $(d \bmod (q^2-1))_{3\ell-1} =
(d \bmod (q^2-1))_{3\ell-2} = 0$, which we need
in Case 3.

\paragraph{Case 3.} $e \ge q$. We induct on the $(q,q+1)$-weighted
degree $qd + (q+1)e$.
In this case, after reducing the $y$-degree
by one step, the monomial reduces to $x^{d+q+1}y^{e-q} - x^dy^{e-q+1}$.
The latter monomial has strictly smaller $(q,q+1)$-weighted degree,
so by induction it is in $\calL$. Thus it suffices to deal with
$x^{d+q+1} y^{e-q}$. Repeating this reduction and ignoring the monomials
with strictly smaller $(q,q+1)$-weighted degree, after at most $m$ reductions
(since $e_k < q$ and so $e < mq$) we have $x^{d+u(q+1)} y^{e \bmod q}$ for
some $u \le m$, which further reduces to
$x^{d + u(q+1) \bmod (q^2-1)}y^{e \bmod q}$. This is almost Case 2,
except for the additional $u(q+1)$ in the exponent of $x$.
By Case 2,
$(d \bmod (q^2-1))_{s+b-1} = (d \bmod (q^2-1))_{s+b} = 0$
and $(e \bmod q)_{s+b} = 0$. Note that since $\ceil{\log_2 m}
< \ell-c$, $u(q+1) \le m(q+1) < 2^{2\ell-c} + 2^{\ell-c} < 2^{s+1}$.
Write $d \bmod (q^2-1)$ as $d' + 2^{s+b+1}d''$ where
$d' < 2^{s+b-1}$. Then
$d + u(q+1) \bmod (q^2-1) = d' + u(q+1) + 2^{s+b+1}d''
< 2^{s+b-1} + 2^{s+1} + 2^{s+b+1}d'' < 2^{s+b} + 2^{s+b+1}d''$
so $(d + u(q+1) \bmod (q^2-1))_{s+b} = 0$.
\end{proof}

\begin{lemma}
\label{lemma:goodsize}
Let ${\rm Good}$ be defined as in Lemma~\ref{lemma:goodcond}.
Let $b = 2+\ceil{\log_2 m}$.
Then
$$
|{\rm Good}| \ge q^{3m}(1 - (1-2^{-mb})^{c/b}).
$$
\end{lemma}
\begin{proof}
We show the equivalent assertion that, by picking
$i_1,\ldots,i_m < q^2$ and $j_1,\ldots,j_m < q$ uniformly at random,
the probability that $(i_1,\ldots,i_m,j_1,\ldots,j_m) \in {\rm Good}$
is $1 - (1 - 2^{-mb})^{c/b})$ at least. Note
that each $j_k < q$ so we only need to consider the $i_k$.
Partition $[3\ell-c,3\ell-1]$ into $c/b$ intervals each of length $b$.
Let $E_i$ be the event that $(i_k)_t = 0$ for all $k \in [m]$
and all $t$ in the $i$th interval. By Lemma~\ref{lemma:goodcond},
if $\bigvee_i E_i$ then $(i_1,\ldots,i_m,j_1,\ldots,j_m) \in {\rm Good}$,
so the probability of landing in ${\rm Good}$ is at least
$$
\Pr\left[ \bigvee_i E_i \right] =
1 - \Pr\left[ \bigwedge_i \overline{E_i} \right] =
1 - (1 - 2^{-mb})^{c/b}.
$$
\end{proof}

Putting together Lemmas~\ref{lemma:goodcond} and~\ref{lemma:goodsize}
with the discussion above, we immediately obtain the following.

\begin{theorem}
\label{theorem:hermliftrate}
Let $m \ge 1$,
let $c > 0$ and let $\delta = 2^{-c}$.
Let $q$ be a power of $2$ such that $\delta q > m$,
and let $r = (1-\delta)q^3$.
Let $\Phi$ be the group of automorphisms on the Hermitian curve
$H \subseteq \F_{q^2}^2$ of the form $(x,y) \mapsto (ax+b, a^{q+1}y
+ab^qx + c)$ and let $\calL = \lift_\Phi^m(\Herm_q[r])$. 
Let $b = 2 + \ceil{\log_2 m}$. Then
the rate of $\calL$ is at least
$1 - (1 - 2^{-mb})^{c/b} \ge 1 - e^{-c/(b2^{mb})}$.
\end{theorem}

Putting everything together, we now prove Theorem~\ref{theorem:main}.

\begin{proof}[Proof of Theorem~\ref{theorem:main}]
Fix $\epsilon, \alpha, N_0 > 0$. Recall that we want, for infinitely many
$N \ge N_0$, a code of length $N$, rate $1-\alpha$, alphabet size
$N^{\epsilon/3}$, and is $(N^\epsilon, \Omega(1))$-locally correctable.

Set $m = \ceil{2/\epsilon}$. Let
$b = 2 + \ceil{\log_2 m}$ and set $c \ge b \cdot 2^{mb} \ln \frac{1}{\alpha}$.
Let $\delta = 2^{-c}$,
set $q$ to be a power of $2$ such that $\delta q > m$ and
$q^{3m} \ge N_0$. Set $N = q^{3m}$ and set $r = (1-\delta)q^3$.
Let $\calL = \lift_\Phi^m(\Herm_q[r])$ where $\Phi$ is the usual automorphism
group of the Hermitian curve $H \subseteq \F_{q^2}^2$.
By our choice of parameters and Theorem~\ref{theorem:hermliftrate},
$\calL$ has block length $q^{3m} = N$,
rate at least $1 - e^{-c/b2^{mb}} \ge 1-\alpha$,
alphabet size $q^2 \le N^{\epsilon/3}$,
has query complexity $q^6 \le N^\epsilon$,
and can correct up to
$\delta^2 = \alpha^{O((8/\epsilon)^{(2/\epsilon)}\log(1/\epsilon))}$.

\end{proof}

\paragraph{Explicitness of code.}
Although a lifted code is not a priori explicit even if the base code is,
Lemma~\ref{lemma:goodcond}
shows that the lifted Hermitian code (more accurately, a subcode
with the same parameter guarantees) is explicit in the following way.
Let ${\rm Good}$ be defined as in Lemma~\ref{lemma:goodcond}.
The $\F_{q^2}$-span of monomials in ${\rm Good}$ have the same rate
guarantees as the full lift, its block length and alphabet size and locality
are the same, and certainly its distance is at least as good, since it
is a subcode. Moreover, to encode a message $m \in \F_{q^2}^{\rm Good}$
into a codeword $\textsf{Enc}(m) \in \F_{q^2}^{H^m}$, first compute
all the monomials
in ${\rm Good}$, which can be done by iterating over every monomial
on $H^m$ and checking if it is in ${\rm Good}$, which can be done
in polynomial time. Then interpret the symbols of $m$ as coefficients
of the monomials in ${\rm Good}$ and let $\textsf{Enc}(m)$ be the evaluations
of $m$ on every point of $H^m$.

\section{Conclusion}
\label{section:conclusion}
In this work, we presented a general framework for constructing
high rate locally correctable codes.
Our framework is an abstraction of affine lifting~\cite{GKS13},
automorphic lifting~\cite{BGKKS13}, and high-degree lifting~\cite{BGKKS13}.
We showed that the lift of a code with good distance
with respect to some $\Phi$ that is close to doubly transitive also
has good distance, and moreover this holds even when the base code
is not invariant under $\Phi$ or when $\Phi$ is not a group.
We showed how one can generalize the construction of the lifted Reed-Solomon
code of~\cite{GKS13} to lift other algebraic geometry codes,
such as the Hermitian code
to obtain locally correctable codes that can attain
query complexity $N^\epsilon$ and rate $1-\alpha$ while
correcting a constant fraction of errors, for any given $\epsilon,\alpha > 0$.

We believe the lifting framework deserves further study.
Lifted codes naturally
have good locality properties. A natural direction to explore is the
local \emph{testability} of lifted codes.
A \emph{local tester} is given oracle access to a word $f$
and must distinguish whether $f \in \calC$ or $\delta(f,\calC) > \epsilon$
for some given constant $\epsilon > 0$. The work of~\cite{GKS13}
shows that affine lifting naturally yields affine-invariant locally
testable codes. An interesting question is whether lifting
algebraic geometry codes yields locally testable codes, and what kind
of assumptions on $\Phi$ are necessary (for example, that the base code
is $\Phi$-invariant or that $\Phi$ is a group).
In fact, \cite{GKS13} shows that both local correctability and
local testability follows generically from affine lifting.
In our work, local correctability follows generically from
lifting --- the instantiation of algebraic geometric base codes
is only used to analyze the rate. It would be interesting to
see if local testability follows generically from lifting as well.

\section*{Acknowledgements}
Thanks to Ronald Cramer and Chaoping Xing for pointing out a bug in a proof
analyzing the Hermitian tower code, which rendered the relevant sections
obsolete.
Thanks to Eli Ben-Sasson for helpful discussions and in particular
for suggesting Theorem~\ref{theorem:dist2}. Thanks to Madhu
Sudan for his helpful comments and encouragement.
Thanks to Badih Ghazi and Henry Yuen and for carefully proofreading parts of
this paper and for helpful comments on the writing.

\bibliographystyle{alpha}
\bibliography{coding}

\end{document}